\documentclass[onecolumn]{IEEEtran}
\IEEEoverridecommandlockouts
\usepackage{cite}
\usepackage{amssymb,amsfonts}

\usepackage{graphicx}
\usepackage{textcomp}
\usepackage[tbtags]{amsmath}
\usepackage{diagbox}
\usepackage{booktabs}
\usepackage{multirow}
\usepackage{amsthm} 
\usepackage{algorithm}
\usepackage{algpseudocode}
\usepackage{color}
\alglanguage{pseudocode}
\algblockdefx{Process}{EndProcess}[1]{{\bf process} #1 }{{\bf end process}}
\algblockdefx{Parallel}{EndParallel}[1]{{\bf begin parallel processing} #1 }{{\bf end parallel processing}}
\algblockdefx{FP}{EndFP}[1]%
{\textbf{for all} #1 \textbf{do in parallel}}%
{\textbf{end for}}

\newtheorem{Theorem}{Theorem}

\theoremstyle{definition}

\def\BibTeX{{\rm B\kern-.05em{\sc i\kern-.025em b}\kern-.08em
    T\kern-.1667em\lower.7ex\hbox{E}\kern-.125emX}}
\begin{document}

\title{Centralized Coded Caching with User Cooperation}

\author{
	\IEEEauthorblockN{Jiahui Chen\IEEEauthorrefmark{1}, Haoyu Yin\IEEEauthorrefmark{1}, Xiaowen You\IEEEauthorrefmark{1}, Yanlin Geng\IEEEauthorrefmark{2} and Youlong Wu\IEEEauthorrefmark{1}\\}  
	\IEEEauthorblockA{\IEEEauthorrefmark{1}School of Information Science and Technology, ShanghaiTech University, Shanghai, China}
	\IEEEauthorblockA{\IEEEauthorrefmark{2}State Key Lab. of ISN, Xidian University, Xi'an, China}
	\{chenjh1, yinhy, youxw\}@shanghaitech.edu.cn, gengyanlin@gmail.com, wuyl1@shanghaitech.edu.cn
	
}

\maketitle

\begin{abstract}
        In this paper, we consider the coded-caching broadcast network with user cooperation, where a server connects with multiple users and the users can cooperate with each other through a cooperation network. We propose a  centralized coded caching scheme based on  a  new deterministic  placement strategy and a parallel  delivery strategy. It is shown that the new scheme optimally allocate the communication loads on the server and users,   obtaining \emph{cooperation gain} and \emph{parallel gain} that greatly reduces the transmission delay. Furthermore, we show that the  number of users who parallelly  send information  should decrease when the users' caching size increases. In other words,  letting more  users  parallelly send information could be harmful. Finally, we derive a constant multiplicative  gap between the lower bound and upper bound on the transmission delay, which proves that our scheme is order optimal.  
        

\end{abstract}

\begin{IEEEkeywords}
Cache, cooperation, delay
\end{IEEEkeywords}

\section{Introduction}

Due to the continuous growth of traffic in the network and the growing needs for higher Internet speed from users, it's imperative to improve the performance of the network.   One of the promising directions to improve the quality of service is to utilize the cache memories in the network.  In \cite{Centralized} Maddah-Ali and Niesen proposed a novel scheme, namely coded caching scheme, to improve the transmission efficiency of each transmission. It  obtains a global caching gain by creating multicasting opportunities for different users. This new class of caching system has attracted significant interests \cite{Decentralized,Multiserver'16,improvedgap,GaussianChannel'17,intermanage'17, StoreandLate'17}. 

To further improve the quality of service, one can combine caching with user cooperation. It is particularly common and useful in  fog  network \cite{Fog'Overiew}, where  the edge users  can carry out some amount of communication. This can include, for example,  device-to-device (D2D) networks and ad-hoc networks.  In \cite{D2D},  coded caching schemes were proposed for a D2D noiseless network with the absence of server. In  \cite{Malak'18}, the caching problem on a two-user  D2D  wireless network with the presence of a server was studied. In \cite{Pedersen'19}, a maximum distance separable (MDS)  coded caching scheme was proposed  to reduce the communication load in highly-dense wireless networks considering device mobility.


{
	In this paper,  we study a $K$-user ($K\geq 2$) coded-caching broadcast network with user cooperation. In this network, a server connects with all users through a noiseless shared link, and the users can communicate with each other through a noiseless cooperation network. The cooperation network is parameterized  by a positive integer $\alpha_{\max}\in\{1,\ldots,\lfloor K/2 \rfloor\}$, denoting the maximum number of users allowed to parallelly send data  in the cooperation network.   For example, when $\alpha_{\max}$ equals to 1, the cooperation network operates as  a simple shared  link connecting with all users, which is  easy and low-cost  to implement in fog network. The main contributions are summarized below. 
\begin{itemize}
\item We propose a novel coded caching scheme that fully exploits user cooperation and  optimally allocates  communication loads between the server and users. The scheme achieves a \emph{cooperation gain} (offered by the cooperation among the users) and a \emph{parallel gain} (offered by the parallel transmission of the server and multiple users) that greatly reduces the transmission delay.
\item We show that the  number of users parallelly  sending information  should decrease with the increase of the users' caching size. In other words,  letting more  users  parallelly send information could be harmful. When users' caching size is sufficiently large, only one user should be allowed to send information, indicating that the cooperation network can be just a simple shared link  connecting all users. These  insights could guide  the implementation of practical caching-aided communication networks with   user cooperation.
\item Lower bounds on the transmission delay are established. Moreover, we show that our scheme achieves the optimal transmission delay within a constant multiplicative  gap.
\end{itemize}

}


The remainder of this paper is as follows. We first present the system model in Section \ref{Sec_Model}, and summarize our main results in Section \ref{Sec_Results}. Followed are detailed descriptions of the centralized coded caching schemes with user cooperation in Section  \ref{Sec_Schemes}. Section \ref{Sec_Conclusion} concludes this paper.

\begin{figure}
           \centering
	\includegraphics[width=0.49\textwidth,scale=0.2]{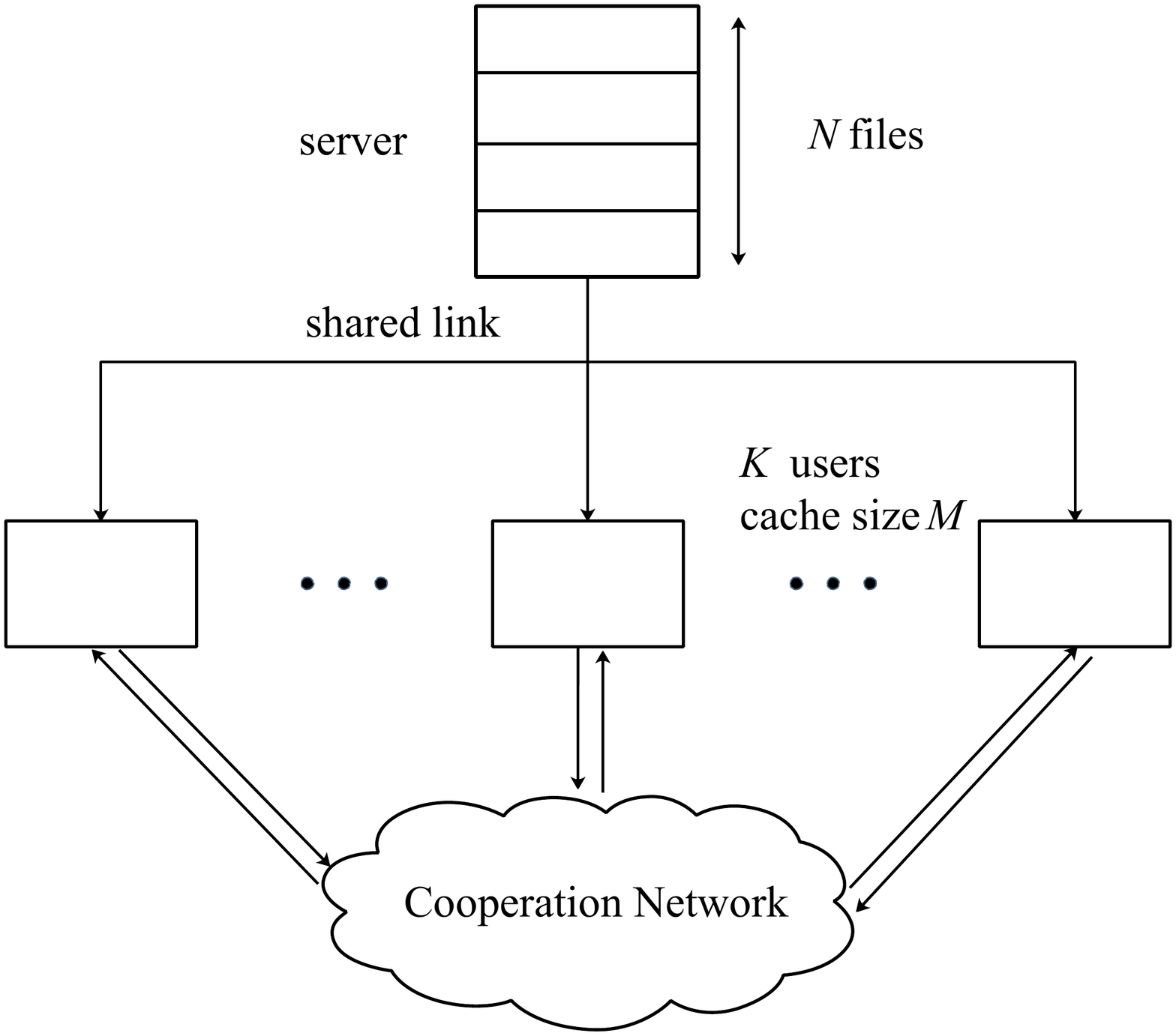}
	\caption{Caching system considered in this paper. A  server connects with $K$ cache-enabled users and the users can cooperate through a flexible network. 
	}
\label{fig_model}
\end{figure}

\section{System Model and Problem Definition}\label{Sec_Model}

Consider a caching network consisting of a single server and $K$ users  as depicted in Fig. \ref{fig_model}.  The server connects with all users through a noiseless shared link, and the users  can communicate  with each other through a cooperation network. 

The server has a database of $N$ independent files $W_1,\ldots,W_N$. Each $W_n$, $n=1,\ldots,N$, is uniformly distributed over \[[2^F]\triangleq\{1,\ldots,2^F\},\] for some positive integer $F$. Each user $k\in[K]$  is equipped with a cache memory of size $MF$ bits, where $M\in [0,N]$.

The cooperation network is parameterized  by a positive integer $\alpha_{\max}\in \mathbb{Z}^+$, denoting the maximum number of users allowed to send data parallelly in the cooperation network. In this work, we made the following assumptions:
\begin{itemize}
\item  The users can receive signals interference-free  from the server and other users simultaneously. This  is reasonable since  the server and the users could operate in two separate bands or layers.  
\item Each user can receive data from at most one user, and if one user sends data, it cannot simultaneously receive data transmitted by other users. This makes  the cooperation network easy to implement and indicates $\alpha_{\max}\leq \lfloor  \frac{K}{2}\rfloor$ .
\item The cooperation network is flexible, in the sense that each user can flexibly select a subset of users to cooperate with, and the subset can be changed during the transmission. This is a similar assumption in \cite{Multiserver'16} and can account for some high-flexibility network such as fog network. 
\end{itemize}
 Given $\alpha_{\max}$ and the assumptions above, the cooperation network can be characterized as below: Let parameter  $\alpha\in \mathbb{Z}^+$, $\alpha\leq \alpha_{\max}$, denote the number of users that exactly  send information parallelly during the data transmission. There exists  a routing strategy at network nodes such that the   $K$ users can be  partitioned into $\alpha$ groups: $\mathcal{G}_1,\ldots,\mathcal{G}_\alpha$, where 
		\begin{IEEEeqnarray}{rCl}
		\mathcal{G}_j\subseteq[K], ~ \mathcal{G}_j\cap \mathcal{G}_i=\emptyset, ~\mathcal{G}_1\cup\cdots \cup\mathcal{G}_{\alpha}=[K].
\end{IEEEeqnarray}
In each group $\mathcal{G}_j$, $j=1,\ldots,\alpha$, a user $u_j$ is allowed to send data to all other users in $\mathcal{G}_j$. The group $\{\mathcal{G}_j\}$ and the user $u_j$ can be  changed during the delivery phase.
	Notice that when  $\alpha_\text{max}=1$, the flexible network degenerates into a  fixed shared link, allowing only  one user to send data to all other users, which is more common and easier to  implement.




The system  works in two phases: a placement phase and a delivery phase. In the placement phase, all users will access the entire library $W_1,\ldots,W_N$ and fill the content to their caches. More specifically,  each user $k$ maps  $W_1,\ldots,W_N$ to its cache contents: 
\begin{IEEEeqnarray}{rCl}
Z_k \triangleq \phi_k(W_1,\ldots,W_N),
\end{IEEEeqnarray}
for some caching function
\begin{IEEEeqnarray}{rCl}\label{eq:caching}
\phi_{k}: [2^F]^N\rightarrow [\lfloor2^{MF}\rfloor].
\end{IEEEeqnarray}

 In the delivery phase,  each user  requests one of the $N$ files from the library. We denote the demand of  user $k$ by 
$d_k\in[N]$,
and the corresponding file  by  $W_{d_k}$. Let 
$\mathbf{d}\triangleq(d_1,\ldots,d_k)$ denote the users' request vector. After the users' requests $\bf{d}$ are informed to  the server and all users,  the server produces symbol
$X\triangleq f_{\bf{d}}(W_1,\ldots,W_N)$,
and user $k\in\{1,\ldots,K\}$ produces symbol\footnote{Each user $k$ can produce $X_k$ as a function of $Z_k$ and the received signals from the server, but since the transmitter's signal is broadcasted to all users, it's equivalent to send $X_k=f(Z_k)$.}
\begin{IEEEeqnarray}{rCl}
X_k\triangleq f_{k,{\bf{d}}}(Z_k),
\end{IEEEeqnarray}
 for some encoding functions
 \begin{subequations}\label{eq:encoding}
\begin{IEEEeqnarray}{rCl}
&&f_{\bf{d}}:[2^F]^N\rightarrow [\lfloor 2^{R_1F} \rfloor],\\
&& f_{k,\bf{d}}: [\lfloor 2^{MF} \rfloor]\rightarrow  [\lfloor 2^{R_2F} \rfloor]\quad \label{eq:userencoding}
\end{IEEEeqnarray}
\end{subequations}
 where $R_1$ and $R_2$ denote the rate transmitted by the server each user, respectively.

 User $k$ perfectly observes the signals sent by the server and other users,  and decodes its desired message as
\[\hat{W}_{d_k}=\psi_{k,\bf{d}}(X,Y_k,Z_k)\]
 where $Y_k\in\{X_1,\ldots,X_K\}$ denotes user $k$'s received signal sent from other users, and  $\psi_{k,\bf{d}}$ is some decoding function 
 \begin{IEEEeqnarray}{rCl}\label{eq:decoding}
 \psi_{k,\bf{d}}: [\lfloor 2^{R_1F} \rfloor]\times [\lfloor 2^{R_2F} \rfloor] \times[ \lfloor 2^{MF} \rfloor]\rightarrow [2^F].
\end{IEEEeqnarray}

 We define the worst-case probability of error as
\begin{IEEEeqnarray}{rCl}
P_e \triangleq \max_{{\bf{d}}\in \mathcal{F}^n} \max_{k\in[K]} \text{Pr} \left(\hat{W}_{d_k} \neq {W}_{d_k}\right).
\end{IEEEeqnarray}

A caching scheme $(M_1,R_1,R_2)$ consists of caching functions \eqref{eq:caching}, encoding functions \eqref{eq:encoding} and decoding functions \eqref{eq:decoding}. We say that the rate region $(M,R_1,R_2)$  is \emph{achievable} if for every $\epsilon>0$ and every large enough file size $F$, there exists a caching scheme such that $P_e$    is less than $\epsilon$. 

Given achievable region $(M,R_1,R_2)$, we define  the \emph{transmission delay} 
\begin{IEEEeqnarray}{rCl}
R\triangleq \max\{R_1,R_2\},
\end{IEEEeqnarray}
and 
the \emph{optimal}  transmission delay 
$R^*\triangleq \inf\{ R\}$.

Our goal is to design the coded caching scheme such that the transmission delay  is minimized. 
Finally,  in this paper we assume
$K \leq N$. Extending the results to the scenario $K \geq N$ is straightforward, referred to \cite{Centralized}.
 
\section{Main Results}\label{Sec_Results}
Consider the cache-aided network with user cooperation described in Section \ref{Sec_Model}. 


\begin{Theorem}\label{Thrm_UpperBound}
	Let $t\triangleq KM/N\in \mathbb{Z}^+$. For  memory size $M \in\{0,N/K,2N/K,\ldots,N\}$, the optimal transmission delay $R^*$ is upper bounded by $R^*\leq R_\textnormal{C}$, where 
	\begin{IEEEeqnarray}{rCl}\label{eq:achievable_rate}
		 R_\textnormal{C} \triangleq \!\!\min_{1\leq\alpha \leq \alpha_\textnormal{max}}\!\!\! K\Big(1-\frac{M}{N}\Big)\frac{1}{1+t+{\alpha\min\{\lfloor  \frac{K}{\alpha}\rfloor\!-\!1,t\}}}.\quad 
	\end{IEEEeqnarray}
	For general $0 \leq M \leq N$, the lower convex envelope of these points is achievable. 
\end{Theorem}
\begin{proof}
See the scheme in Section \ref{Sec_Schemes}.
\end{proof}

Recall that parameter $\alpha$ denotes the number of users who exactly  send information parallelly in the delivery phase. From \eqref{eq:achievable_rate}, it's easy to obtain the optimal value of $\alpha$, denoted by $\alpha^*$:
\begin{equation} \label{eq:optimalAlpha}
	\alpha^*=\!\left\{
	\begin{aligned}
	&1, &t\geq K-1,\\
	&\max\! \Big\{\!\alpha\!:\! \lfloor\frac{K}{\alpha}\rfloor\!-\!1\!=\!t\Big\}, \!&\lfloor  \frac{K}{\alpha_\textnormal{max}}\rfloor\!-\!1 \!<\!t\!<\! K\!-\!1,\\
	&\alpha_\textnormal{max}, & t \leq \lfloor  \frac{K}{\alpha_\textnormal{max}}\rfloor\!-\!1.
	\end{aligned}
	\right.
\end{equation}
It's interesting to find that the  number of users parallelly  sending information should decrease as  the users' caching size $M$ increases for given  $(K, N, \alpha_{\max})$. To simplify the explanation, we assume $\alpha_\textnormal{max}=\lfloor \frac{K}{2}\rfloor$ and  $KM/N\in \mathbb{Z}^+$.
When $M\leq N(\lfloor  \frac{K}{\alpha_\textnormal{max}}\rfloor-1)/K$, we have $\alpha^*=\alpha_\textnormal{max}$ and thus it's beneficial to let the most   users  parallelly send information.  As  $M$ increases,  $\alpha^*$ decreases and $\alpha^*< \alpha_\textnormal{max}$, which indicates that  letting more  users  parallelly send information could be harmful. For example, when $N=100$, $K=10$, $\alpha_\textnormal{max}=5$, $M=40$, then $\alpha^*=2<\alpha_\textnormal{max}$, and thus only two users rather than 5 users should parallelly send data.  In the extreme case when $M\geq (K-1)N/K$,  only one user should be allowed to send information,   implying that when users' caching size is sufficiently large, the  cooperation network can be just a simple share link  connecting all users. 

	
	
	
 Comparing $R_\textnormal{C}$ with the delay achieved by the  scheme without user cooperation in \textnormal{\cite{Centralized}}, i.e., $K\big(1-\frac{M}{N}\big)\frac{1}{1+t}$,    $R_\textnormal{C}$ consists of an additional factor \[G_\text{c}\triangleq\frac{1}{1+{\frac{\alpha^*}{1+t}\min\{\lfloor  \frac{K}{\alpha^*}\rfloor\!-\!1,t\}}},\] which we call the \emph{cooperation gain},  since it arises  from the user cooperation.  Comparing  $R_\textnormal{C}$  with  the delay achieved by the  scheme for D2D network without server \cite{D2D}, i.e.,  $\frac{N}{M}(1-\frac{M}{N})$,   $R_\textnormal{C}$ consists of an additional factor 
	\[ G_\text{p} \triangleq\frac{1}{1+\frac{1}{t}+\frac{\alpha^*}{t}\min\{\lfloor  \frac{K}{\alpha^*}\rfloor\!-\!1,t\}},
	\]
 which we call the \emph{parallel gain},  since it arises  from  the parallel transmission of the server and users. Both gains  depend on $K$, $M/N$ and $\alpha_{\max}$. When fixing $(K, N, \alpha_{\max})$, $G_\text{c}$ in general is not a  monotonic function of $M$. It is monotonic decreasing when $M\leq \lfloor  \frac{K}{\alpha_\textnormal{max}}\rfloor\!-\!1$, and monotonic decreasing when $M\geq \lfloor  \frac{K}{\alpha_\textnormal{max}}\rfloor\!-\!1$.  $G_\text{p}$ is  a  monotonic increasing function of $M$.   The cooperation gain and parallel gain are plotted in  Fig. \ref{fig_gain} for given  $(K, N, \alpha_{\max})$.  
 \begin{figure}
           \centering
	\includegraphics[width=0.49\textwidth,scale=0.2]{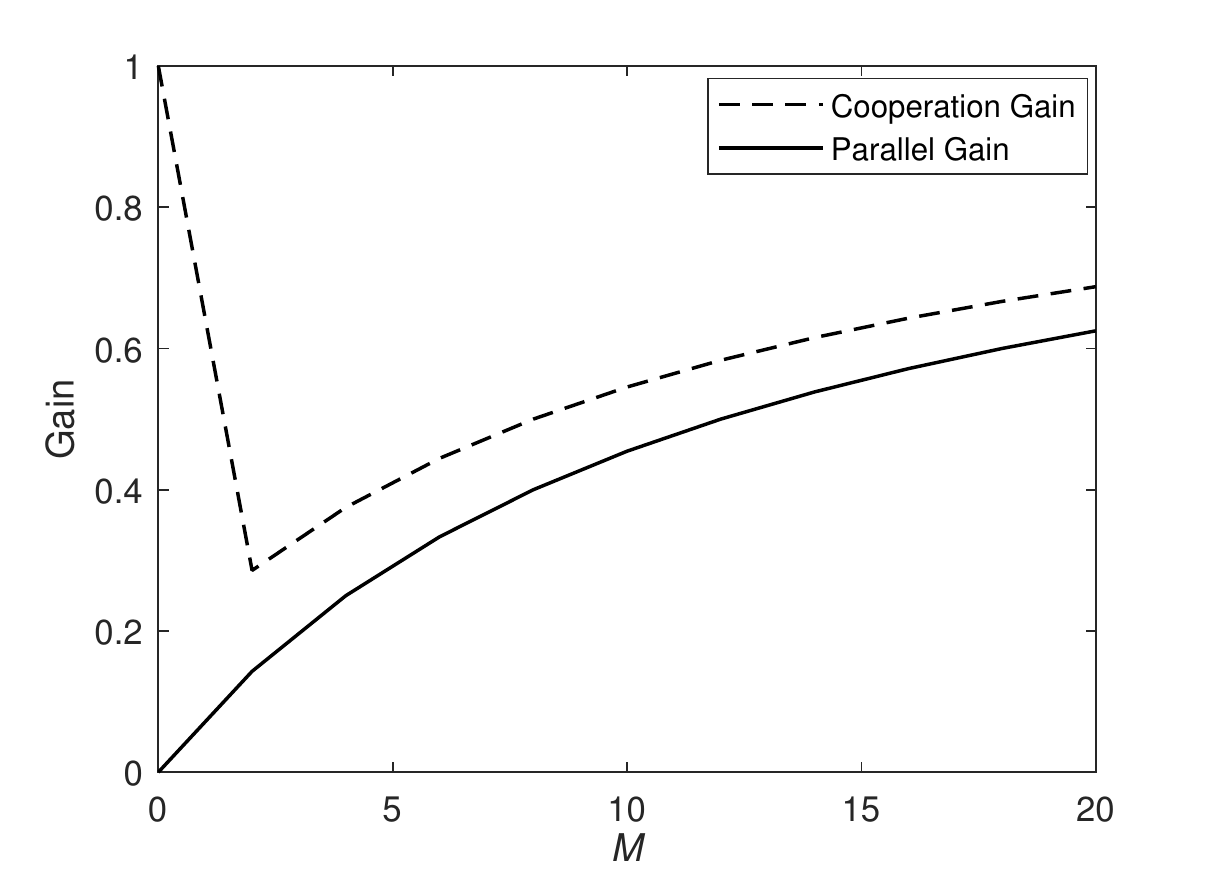}
	\caption{Cooperation gain and parallel gain  when $N=20$, $K=10$ and $\alpha_{\max}=5$. 
	}
\label{fig_gain}
\end{figure}

\begin{Theorem}\label{Thrm_LowerBound}
	For  memory size $0 \leq M \leq N$, the optimal transmission delay is lower bounded by
	\begin{IEEEeqnarray}{rCl}\label{eq:cutset}
		R^*\geq \max&&\left\{\frac{1}{2}\Big(1-\frac{M}{N}\Big),\max\limits_{s\in [K]}\Big(s-\frac{KM}{\lfloor N/s\rfloor}\Big), \right.\nonumber\\&&\qquad\left.\max\limits_{s\in [K]}\Big(s-\frac{sM}{\lfloor N/s\rfloor}\Big)\frac{1}{1+\alpha_\textnormal{max}}\right\}.\label{eq:cutset2}
	\end{IEEEeqnarray}
\end{Theorem}
\begin{proof}
See the proof in Appendix \ref{App_converse}.
\end{proof}

\begin{Theorem}\label{Thrm_Gap}
	For  memory size $0 \leq M \leq N$, 
	\begin{IEEEeqnarray}{rCl}
		R_\textnormal{C}  /R^*\leq 31.
	\end{IEEEeqnarray}
\end{Theorem}
\begin{proof}
See the proof in Appendix \ref{App_Gap}.
\end{proof}

\begin{figure}
	\centering
	\includegraphics[width=0.49\textwidth,scale=0.2]{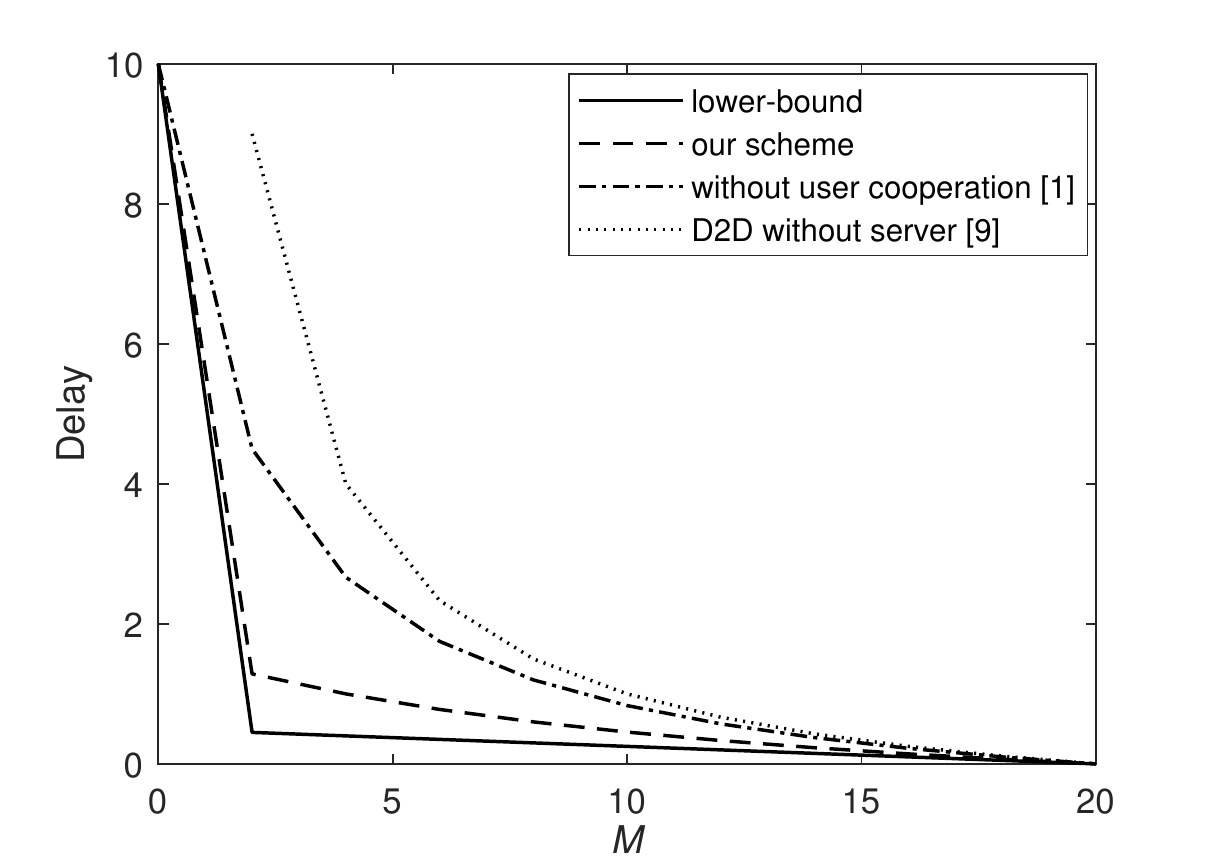}
	\caption{The lower  and upper bounds on the transmission delay  when $N=20$, $K=10$ and $\alpha_{\max}=5$.
	}
	\label{fig_result}
\end{figure}

The numerical result is presented in Fig. \ref{fig_result}. It shows that the delay of the proposed scheme is  close to the lower bound and much tighter than the upper bound achieved by the the schemes in \cite{Centralized} and \cite{D2D}.

\section{Proof of Theorem \ref{Thrm_UpperBound}} \label{Sec_Schemes}
In this section, we present a caching scheme  achieving an upper bound on the transmission delay for the network  depicted in Fig. \ref{fig_model}.  One may come up with the idea to use the  scheme introduced in \cite{D2D} considering caching-aided D2D network without server, but there are two main problems:  
 \begin{itemize}
 \item When each user sends data to other users in its partition group,   in order to achieve the maximum multicast gain, user $u_j$  in the  group $\mathcal{G}_j$ should broadcast a coded data consisting of $|\mathcal{G}_j|\!-\!1$ useful subfiles required by the remaining users in  group $\mathcal{G}_j$. Also, the amount of subfiles should support  the server and $\alpha$ users  to simultaneously send data in every transmission slot. These can not be guaranteed by the file-splitting process and caching placement phase introduced in \cite{D2D}.
\item In  \cite{D2D}, the users are fixed in a mesh network, leading to an unchanging group partition during the delivery phase, which is not the same case in our model. Moreover,  our model can have the server  share some communication loads with the users. These two facts result in great difference in the delivery phase compared to that  in \cite{D2D}. To achieve the optimal delay, we need to fully exploit the  dynamic strategy and  optimally allocate the communicate loads at the server and users.  
 \end{itemize}
 
We describe a novel coded caching scheme for any $K$, $N$ and $M$ such that $t= \frac{KM}{N}$ is a positive integer.  When $t$ is not an integer, 
 we can use a resource sharing scheme as in  \cite{Centralized}.

 Introduce integers $\alpha$, $L_1,L_2$ and $L\triangleq L_1+L_2$, where
$1\leq\alpha\leq \alpha_\text{max}$, $L_2> 0$, and  $L_1\geq 0$ such that 
\begin{subequations}
\begin{IEEEeqnarray}{rCl}\label{eq_Lcondition}
&&\frac{K\cdot\binom{K-1}{t}\cdot L_1}{\alpha\min\{\lfloor  \frac{K}{\alpha}\rfloor-1,t\}}\in \mathbb{Z}^+,
\end{IEEEeqnarray}
and 
\begin{IEEEeqnarray}{rCl}\label{eq_Acondition}
&& \frac{L_1}{L_2}=\frac{{\alpha\min\{\lfloor  \frac{K}{\alpha}\rfloor-1,t\}}}{1+t}.
\end{IEEEeqnarray}
\end{subequations}
 Here $L_1/L$ and $L_2/L$ denote the proportions of communication loads assigned to the users and the server respectively. Condition \eqref{eq_Lcondition} ensures that the number of subfiles can support  a maximum multicast gain when user sending data,  and \eqref{eq_Acondition} ensures that   communication loads can be optimally allocated  at the server and the users.
 


 In the placement phase, each file is split into $L\binom{K}{t}$ subfiles of equal size. We index the subfiles of  $W_n$ by the superscript $l\in [L]$ and subscript $\mathcal{T}\subset [K]$:
 \begin{IEEEeqnarray}{rCl}
 W_n = \left(W_{n,\mathcal{T}}^{l}: l\in[L],\mathcal{T}\subset [K], |\mathcal{T}|=t\right).
\end{IEEEeqnarray}
User $k$ caches all the subfiles when $k \in \mathcal{T}$ for all $n=1,...,N$ and $l=1,...,L$, so it requires
$$N\cdot \frac{F}{L\binom{K}{t}} \cdot L\binom{K-1}{t-1} =F\cdot \frac{Nt}{K}=MF$$
bits of cache, satisfying the cache size constraint. 



In the delivery phase, Each user $k$ requests file $W_{d_k}$. The requests vector $\mathbf{d}$ are informed by  the server and all the users. Note that different parts of the file $W_{d_k}$ has been stored in the users' caches, and  thus the uncached parts of $W_{d_k}$ can   be sent by the server and users. Divide the uncached subfiles of $W_{d_k}$ into two parts: one that is sent by the server and the other that is sent by the users in the network.
Subfiles
 \[\left(W_{d_k,\mathcal{T}}^{1},\ldots,W_{d_k,\mathcal{T}}^{L_1}:\mathcal{T}\subset [K], |\mathcal{T}|=t, k\notin \mathcal{T}\right)\]   are requested  by user $k$ and will be sent by the users, thus  $\frac{L_1}{L}$ represents the fraction of the subfiles sent by the users. Subfiles  \[\left(W_{d_k,\mathcal{T}}^{L_1+1},\ldots,W_{d_k,\mathcal{T}}^{L}:\mathcal{T}\subset [K], |\mathcal{T}|=t, k\notin \mathcal{T}\right)\] are requested  by user $k$  and  will be sent by the server, thus $\frac{L_2}{L}$ represents the the fraction of the subfiles sent by the server.  
 
 Our objective is to get the upper bound of transmission delay in the worst request case, so we assume that each of the users makes unique requests. Thus,  the total number of requested subfiles for $K$ users is $K\cdot\binom{K-1}{t}\cdot L$.

First consider the subfiles sent by the users. In order to create multicast opportunities among users, we partition the $K$ users into $\alpha$ groups of equal size: \[\mathcal{G}_1,\ldots,\mathcal{G}_{\alpha}\] where for $i,j=1,\ldots,\alpha$,
$\mathcal{G}_i\subseteq [K]: |\mathcal{G}_i|={\lfloor {K}/{\alpha} \rfloor}$, and  $\mathcal{G}_i\cap \mathcal{G}_j=\emptyset$,~\textnormal{if $i\neq j$}. In each group $\mathcal{G}_i$, one of $\lfloor {K}/{\alpha} \rfloor$ users plays the role of server and  sends symbols based on its cached contents to the remaining $(\lfloor {K}/{\alpha} \rfloor-1)$ users in the group. 

Focus on a group $\mathcal{G}_i$ and a set $\mathcal{S}\subset[K]:|\mathcal{S}|=t+1$. If $|\mathcal{G}_i|\leq|\mathcal{S}|$, then all nodes in $\mathcal{G}_i$ share  subfiles  $(W^l_{n,\mathcal{T}}:l\in[L_1],n\in[N],\mathcal{G}_i\subseteq\mathcal{T},|\mathcal{T}|=t)$. For this case,  user  $k\in\mathcal{G}_i$ sends a XOR symbol that contains the requested subfiles for all remaining $\lfloor K/\alpha \rfloor -1$ users in $\mathcal{G}_i$.  If $|\mathcal{G}_i|>|\mathcal{S}|$, then the nodes in $\mathcal{S}$ share  subfiles  $(W^l_{n,\mathcal{T}}:l\in[L_1],n\in[N],\mathcal{T}\subset\mathcal{S},|\mathcal{T}|=t)$. For this case,  user $k\in\mathcal{S}$ sends a XOR symbol that contains the requested subfiles for all remaining $t$ users in $\mathcal{S}$.  Other groups perform the similar steps and concurrently deliver the remaining requested subfiles to other users.

By changing the partition and performing the delivery strategy described above, we can finally send all the requested subfiles 
\begin{IEEEeqnarray}{rCl}\label{eq_userDfiles}
(W_{d_k,\mathcal{T}}^{1},\ldots,W_{d_k,\mathcal{T}}^{L_1}:\mathcal{T}\subset [K], |\mathcal{T}|=t, k\notin \mathcal{T})^K_{k=1}
\end{IEEEeqnarray}
 to the  users. Since $\alpha$ groups  work in a parallel manner ($\alpha$ users  can concurrently deliver contents), and  each user in a group delivers a symbol containing $\min\{ \lfloor K/\alpha \rfloor -1, t\}$ subfiles requested by the users in its group, to send all requested subfiles in \eqref{eq_userDfiles}, we need
\begin{IEEEeqnarray}{rCl}\label{eq_centraTimes}
\frac{K\cdot\binom{K-1}{t}\cdot L_1}{\alpha\min\{\lfloor  \frac{K}{\alpha}\rfloor-1,t\}}
\end{IEEEeqnarray}
times of transmission each of rate $\frac{1}{L\binom{K}{t} }$. Notice that  $L_1$ is chosen according to \eqref{eq_Lcondition},  ensuring that   \eqref{eq_centraTimes} equals to an integer. Thus, the  rate sent by the users is 
\begin{IEEEeqnarray}{rCl}\label{eq_R2}
R_2&=&  \frac{K\cdot\binom{K-1}{t}\cdot L_1}{\alpha\min\{\lfloor  \frac{K}{\alpha}\rfloor-1,t\}}\cdot \frac{1}{L\binom{K}{t} }\nonumber\\
	&=&\frac{L_1}{L}\cdot K\left(1-\frac{M}{N}\right)\cdot \frac{1}{{\alpha\min\{\lfloor  \frac{K}{\alpha}\rfloor-1,t\}}}.
\end{IEEEeqnarray}





Now we describe the delivery of the subfiles sent by the server. For each $l=L_1+1,\ldots,L$, apply the  delivery strategy as in \cite{Centralized}. Specifically, the server sends $\oplus_{k\in \mathcal{S}} W_{d_k,\mathcal{S}\backslash\{k\}}^{l}
$
for all   $\mathcal{S}\subseteq[K]:|\mathcal{S}|=t+1$ and $l=L_1+1,\ldots,L$.
We obtain the  rate sent by the server 
\begin{equation} \label{eq_R1}
\begin{aligned}
		R_1=\frac{L_2}{L}\cdot K\left(1-\frac{M}{N}\right)\cdot \frac{1}{1+KM/N}.	 
\end{aligned}
\end{equation}
Since the server and users transmit the signals simultaneously, the transmission delay of the whole network is the maximum between $R_1$ and $R_2$, i.e., $R=\max\{R_1,R_2\}$.

Given \eqref{eq_Acondition}, \eqref{eq_R2} and \eqref{eq_R1}, we find that  $R_1=R_2$, which means that the transmission delay reaches the optimal point.  The transmission delay can thus be rewritten as
\begin{equation*}
	\begin{aligned}
	R
	&=K\left(1-\frac{M}{N}\right)\frac{1}{1+t+{\alpha\min\{\lfloor  \frac{K}{\alpha}\rfloor-1,t\}}}.
	\end{aligned}
\end{equation*}

In order to explain the key steps in caching scheme described above, we will then present a simple example.\\
\textbf{Example 1}: Consider a network consisting of $K=6$ users with cache size $M=4$, and a library of $N = 6$ files. Thus $t=KM/N=4$. Let $\alpha=2$, that is to say we separate the 6 users into 2 groups of equal size. The choice of the groups
is not unique. 
We choose $L_2=1$ and $L_1=2$ such that $\frac{K\binom{K-1}{t}L_1}{\min\{\alpha(\lfloor K/\alpha \rfloor-1),t\}}=15$ is an integer. 
Split each file $W_n$, for $n=1,\ldots,N$, into $3\binom{6}{4}=45$ subfiles:
 \[W_n = (W^l_{n,\mathcal{T}}: l\in[3],\mathcal{T}\subset[6], |\mathcal{T}|=4).\]
 
 We list all the requested subfiles by  the users as follows: for $l=1,2,3$,
\begin{equation*}
\begin{aligned}
&W_{d_1,\{2345\}}^{l},W_{d_1,\{2346\}}^{l},W_{d_1,\{2356\}}^{l},W_{d_1,\{2456\}}^{l},W_{d_1,\{3456\}}^{l}\\ &W_{d_2,\{1345\}}^{l},W_{d_2,\{1346\}}^{l},W_{d_2,\{1356\}}^{l},W_{d_2,\{1456\}}^{l},W_{d_2,\{3456\}}^{l}\\ &W_{d_3,\{1245\}}^{l},W_{d_3,\{1246\}}^{l},W_{d_3,\{1256\}}^{l},W_{d_3,\{1456\}}^{l},W_{d_3,\{2456\}}^{l}\\ &W_{d_4,\{1235\}}^{l},W_{d_4,\{1236\}}^{l},W_{d_4,\{1256\}}^{l},W_{d_4,\{1356\}}^{l},W_{d_4,\{2356\}}^{l}\\ &W_{d_5,\{1234\}}^{l},W_{d_5,\{1236\}}^{l},W_{d_5,\{1246\}}^{l},W_{d_5,\{1346\}}^{l},W_{d_5,\{2346\}}^{l}\\ &W_{d_6,\{1234\}}^{l},W_{d_6,\{1235\}}^{l},W_{d_6,\{1245\}}^{l},W_{d_6,\{1345\}}^{l},W_{d_6,\{2345\}}^{l}.
\end{aligned}
\end{equation*}

 The users can finish the transmission in different partitions.  Table 1 shows one kind of the partition for example and explains how the users send the requested subfiles for $l=1,2$. 
\begin{table}  
\centering
	\caption{Subfiles sent by users in different partition}
	\label{table_time}
	\begin{tabular}{cc}  	
		\toprule
		  $\{1,2,3\}$&$\{4,5,6\}$	\\
		  user 2:  $W_{d_1,\{2345\}}^{1}\oplus W_{d_3,\{1245\}}^{1}$&user 5: $W_{d_4,\{2356\}}^{1}\oplus W_{d_6,\{2345\}}^{1}$\\
		  user 2:  $W_{d_1,\{2346\}}^{1}\oplus W_{d_3,\{1246\}}^{1}$&user 5: $W_{d_4,\{1256\}}^{1}\oplus W_{d_6,\{1245\}}^{1}$\\
		  user 1:  $W_{d_2,\{1346\}}^{1}\oplus W_{d_3,\{1256\}}^{1}$&user 4: $W_{d_5,\{2346\}}^{1}\oplus W_{d_6,\{1345\}}^{1}$\\
		  user 3:  $W_{d_1,\{2356\}}^{1}\oplus W_{d_2,\{1356\}}^{1}$&user 6: $W_{d_4,\{1356\}}^{1}\oplus W_{d_5,\{1346\}}$\\
		  \midrule
		  $\{1,2,4\}$&$\{3,5,6\}$\\
		  user 2: $W_{d_1,\{2456\}}^{l}\oplus W_{d_4,\{1235\}}^{l}$&user 5: $W_{d_3,\{1456\}}^{l}\oplus W_{d_6,\{1235\}}^{l}$\\
\midrule
		  $\{1,4,6\}$&$\{2,3,5\}$\\
		  user 6:  $W_{d_1,\{3456\}}^{l}\oplus W_{d_4,\{1236\}}^{l}$&user 3: $W_{d_2,\{3456\}}^{l}\oplus W_{d_5,\{1234\}}^{l}$\\
\midrule	
		  $\{1,2,5\}$&$\{3,4,6\}$\\
		  user 1: $W_{d_2,\{1456\}}^{l}\oplus W_{d_5,\{1236\}}^{l}$&user 4:  $W_{d_3,\{2456\}}^{l}\oplus W_{d_6,\{1234\}}^{l}$\\
\midrule
		  $\{1,2,3\}$&$\{4,5,6\}$\\
		  user 3: $W_{d_1,\{2345\}}^{2}\oplus W_{d_2,\{1345\}}^{2}$&user 4: $W_{d_5,\{2346\}}^{2}\oplus W_{d_6,\{2345\}}^{2}$\\
		  user 3: $W_{d_1,\{2346\}}^{2}\oplus W_{d_2,\{1346\}}^{2}$&user 4: $W_{d_5,\{1246\}}^{2}\oplus W_{d_6,\{1245\}}^{2}$\\		  
		  user 2: $W_{d_1,\{2356\}}^{2}\oplus W_{d_3,\{1245\}}^{2}$&user 5: $W_{d_4,\{1356\}}^{2}\oplus W_{d_6,\{1345\}}^{2}$\\
		  user 1: $W_{d_3,\{1246\}}^{2}\oplus W_{d_2,\{1356\}}^{2}$&user 6: $W_{d_4,\{1256\}}^{2}\oplus W_{d_5,\{1346\}}^{2}$\\
		  user 1: $W_{d_3,\{1256\}}^{2}\oplus W_{d_2,\{1345\}}^{1}$&user 6: $W_{d_5,\{1246\}}^{1}\oplus W_{d_4,\{2356\}}^{2}$\\
		\bottomrule  
	\end{tabular}
\end{table}
In Table \ref{table_time}, all the users send an XOR symbol of subfiles with superscript  $l=1$  at the beginning. Note that the subfiles  $W_{d_2,\{1345\}}^1$ and $W_{d_5,\{1246\}}^1$ are left since $\frac{K\binom{K-1}{t}}{\alpha(\lfloor K/\alpha \rfloor-1)}$ is not an integer. Similarly, for subfiles with  $l=2$,   $W_{d_3,1256}^2$ and $W_{d_4,2356}^2$ are not sent to user 3 and 4. In the last transmission,  user 1 delivers the XOR message $W_{d_3,\{1256\}}^{2}\oplus W_{d_2,\{1345\}}^{1}$  to user 2 and 3, and user 6 delivers $W_{d_5,\{1246\}}^{1}\oplus W_{d_4,\{2356\}}^{2}$ to user 5 and 6. The  rate transmitted by the users is $R_2=\frac{1}{3}.$

  The server delivers subfiles for $l=3$   in the same way as in \cite{Centralized}. Specifically, it sends symbols $\oplus_{k\in \mathcal{S}} W_{d_k,\mathcal{S}\backslash\{k\}}^{3}$,
for all $\mathcal{S}\subseteq[K]:|\mathcal{S}|=5$. Thus the rate sent by the server is $R_1=\frac{2}{15}$,
and the transmission delay $R=\max\{R_1,R_2\}=\frac{1}{3}$,
which is less than the delay achieved by the centralized coded caching scheme without user cooperation.

\section{Conclusions}\label{Sec_Conclusion}
 In this paper, we consider a coded-caching broadcast network with user cooperation. An order optimal  scheme is proposed which achieves a cooperation gain and a parallel gain, by exploiting user cooperation and parallel transmission at the server and users. Furthermore, we show that  letting more  users  parallelly send information could be harmful. The directions of  future research  could be on  the decentralized caching problem  and wireless network with caching and cooperation.
\appendices

\section{Converse Proof}\label{App_converse}

	
	Due to the flexibility of  cooperation network, the connection and partitioning status between users can change during the delivery phase, we can't drive our converse directly like \cite{Centralized}.  Moreover, the parallel transmission of the server and many users results in abundant transmitting signals, making the scenario more sophisticated. 

Let $R^*_1$ and $R^*_2$ denote the optimal rate sent by the server and each user. We first consider an ideal case where every user is served by a exclusive server and  user, which both store full files in the database, then  we easy to obtain $R^*\geq \frac{1}{2}(1-\frac{M}{N}).$ 

 Next,  consider the first $s$ users with cache contents $Z_1,...,Z_s$. Define $X_{1,0}$ to be the signal sent by the server, and $X_{{1,1}},\ldots,X_{{1,\alpha_\text{max}}}$ to be the signals sent by the $\alpha_\text{max}$ users, respectively,  where $X_{{j,i}}\in[\lfloor 2^{R^*_2F} \rfloor]$ for $j\in[s]$ and $i\in[\alpha_\text{max}]$. Assume that $W_1,\ldots,W_s$ is determined by  $X_{1,0}$, $X_{{1,1}},\ldots,X_{{1,\alpha_\text{max}}}$ and  $Z_1,\ldots,Z_s$. Also,  define $X_{2,0}$, $X_{{2,1}},\ldots,X_{{2,\alpha_\text{max}}}$ to be the signals  which enable the users to decode $W_{s+1},...,W_{2s}$. Continue the same process such that $X_{\lfloor N/s\rfloor,0}$, $X_{{\lfloor N/s\rfloor,1}},\ldots,X_{{\lfloor N/s\rfloor,\alpha_\text{max}}}$  are the signals which enable the users to decode $W_{s\lfloor N/s\rfloor-s+1},...,W_{s\lfloor N/s\rfloor}$. We then have $Z_1,\ldots,Z_s$,  $X_{1,0},\ldots,X_{\lfloor N/s\rfloor,0}$, and $X_{{1,1}},  \ldots,X_{{1,\alpha_\text{max}}},\ldots, X_{{\lfloor N/s\rfloor,1}},\ldots,X_{{\lfloor N/s\rfloor,\alpha_\text{max}}}$   determine $W_{1},\ldots,W_{s\lfloor N/s\rfloor}$. Let
\[ 
 {\bf{X}}_{1:\alpha_\text{max}} \triangleq (X_{{1,1}},  \ldots,X_{{1,\alpha_\text{max}}},\ldots, X_{{\lfloor N/s\rfloor,1}},\ldots,X_{{\lfloor N/s\rfloor,\alpha_\text{max}}}). 
\]
By the definitions of $R^*_1$,  $R^*_2$ and the encoding function \eqref{eq:userencoding}, we have 
\begin{subequations}\label{eq_entropy}
	\begin{IEEEeqnarray}{rCl}
		&&H(X_{1,0},\ldots,X_{\lfloor  {N}/{s}\rfloor,0}) \leq \lfloor  {N}/{s}\rfloor R^*_1F,\\
		&&H({\bf{X}}_{1:\alpha_\text{max}}) \leq \lfloor {N}/{s}\rfloor \alpha_\text{max} R^*_2F, \\
		&&H({\bf{X}}_{1:\alpha_\text{max}},Z_1,\ldots,Z_s) \leq  KMF.
	\end{IEEEeqnarray}
\end{subequations}

Consider then the cut separating $X_{1,0},\ldots,X_{\lfloor N/s\rfloor,0}$, ${\bf{X}}_{1:\alpha_\text{max}}$, and $Z_1,\ldots,Z_s$ from the corresponding $s$ users. By the cut-set bound and \eqref{eq_entropy}, we have
\begin{IEEEeqnarray}{rCl}
	\lfloor \frac{N}{s}\rfloor sF &\leq&\lfloor \frac{N}{s}\rfloor R^*_1F+KMF,\\
	\lfloor \frac{N}{s}\rfloor sF   & \leq&  \lfloor \frac{N}{s}\rfloor R^*_1F+sMF+\lfloor \frac{N}{s} \rfloor \alpha_\text{max} R^*_2F.
\end{IEEEeqnarray}
Since we have $R^*\geq R^*_1$ and $R^*\geq \max\{R^*_1,R^*_2\}$ from the above definition, solving for $R^*$ and optimizing over all possible choices of $s$, we obtain
\begin{subequations}\label{eq_proofCutset}
	\begin{IEEEeqnarray}{rCl}
		R^*&\geq &\max\limits_{s\in [K]}(s-\frac{KM}{\lfloor N/s\rfloor}),\\
		R^*&\geq& \max\limits_{s\in [K]}(s-\frac{sM}{\lfloor N/s\rfloor})\frac{1}{1+\alpha_\text{max}}.
	\end{IEEEeqnarray}
\end{subequations}

\section{Proof of Theorem \ref{Thrm_Gap}}\label{App_Gap}

We prove that  $R_\text{C}$  is within a constant multiplicative gap of the minimum feasible delay $R^*$ for all values of $M$. To prove the result, we
compare  them in the  following regimes.
\begin{itemize}
\item 
If $ 0.6393 < t<  \lfloor  K/\alpha\rfloor-1$, from Theorem \ref{Thrm_LowerBound}, we have 
\begin{equation}\label{eq_as}
	\begin{aligned}
	R^*&\geq(s-\frac{Ms}{\lfloor N/s\rfloor})\frac{1}{1+\alpha_\text{max}}\\
	   &\overset{(a)}\geq \frac{1}{12}\cdot K\Big(1-\frac{M}{N}\Big)\frac{1}{1+t}\cdot\frac{1}{1+\alpha_\text{max}},
	\end{aligned}
\end{equation}
where (a) follows from   \cite[Theorem 3]{Centralized}. Then we have
\begin{IEEEeqnarray}{rCl}
\frac{R_\text{C}}{R^*}
&\leq& 12\cdot\frac{(1+\alpha_\text{max})(1+t)}{1+t+\alpha t}\nonumber\\
	&=&12\cdot\frac{(1+\alpha_\text{max})}{1+\alpha t/(1+t)}\nonumber\\
	&\leq& 12\cdot\frac{(1+\alpha_\text{max})}{1+\alpha\cdot 0.6393/(1+0.6393)}\nonumber\\
	&\leq &31
\end{IEEEeqnarray}
where the last inequality holds since we can choose $\alpha=\alpha_\text{max}$. 
\item If $t> \lfloor  K/\alpha\rfloor-1  $, we have 
\begin{IEEEeqnarray}{rCl}
\frac{R_\text{C}}{R^*} &\leq  &\frac{K(1-\frac{M}{N})\frac{1}{1+t+\alpha (\lfloor K/\alpha \rfloor-1)}}{\frac{1}{2}(1-\frac{M}{N})}\nonumber\\
&= &\frac{2K}{1+t+\alpha (\lfloor K/\alpha \rfloor-1)}\nonumber\\
&\overset{(a)}\leq  &\frac{2K}{K+KM/N}\nonumber\\
&\leq& 2
\end{IEEEeqnarray}
where $(a)$ follows from that  we can  choose $\alpha=1$.

\item If $t\leq 0.6393$, 
setting $s=0.275N$, we have
\begin{IEEEeqnarray}{rCl}\label{eq_lower}
R^* &\geq& s-\frac{KM}{\lfloor N/s\rfloor}\nonumber\\
			&\overset{(a)}\geq& s-\frac{KM}{N/s-1}\nonumber\\
			&=&0.275N-t\cdot0.3793N\nonumber\\
			&\geq& 0.0325N> \frac{1}{31}\cdot N
\end{IEEEeqnarray}

where $(a)$ holds since $\lfloor x \rfloor \geq x-1$ for any $x \geq 1$. 
Note that for all values of $M$, the transmission delay 
\begin{equation}\label{upper}
	R_\text{C}\leq \min \{K, N\}. 
\end{equation}
Combining with \eqref{eq_lower} and \eqref{upper}, we have ${R_\text{C}}/{R^*} \leq 31$.

\end{itemize}


\begin{thebibliography}{00}
\bibitem{Centralized} M. A. Maddah-Ali and U. Niesen, ``Fundamental limits of caching,'' \textit{IEEE Trans. Info. Theory,} vol. 60, no. 5, pp. 2856--1867, May 2014.

\bibitem{Decentralized} M. A. Maddah-Ali and U. Niesen, ``Decentralized coded caching attains order-optimal memory-rate tradeoff,'' \textit{IEEE/ACM Trans. on Networking,} vol. 23, no. 4, pp. 1029--1040, Aug. 2015.



\bibitem{Multiserver'16} S. P. Shariatpanahi, S. A. Motahari, and B. H. Khalaj, ``Multi-server coded caching,'' \textit{ IEEE Trans. on Info Theory,} vol. 62, no. 12, pp. 7253--7271, Dec. 2016.

\bibitem{improvedgap}
C. Wang, S. Saeedi Bidokhti and M. Wigger, ``Improved converses and gap results for coded caching,''  \emph{IEEE Trans. Info. Theory}, vol. 64, no. 11, pp. 7051-7062, Nov. 2018.


\bibitem{GaussianChannel'17}
S. S. Bidokhti, M. Wigger and A. Yener, ``Gaussian broadcast channels with receiver cache assignment,'' in \textit{IEEE International Conference on Communications (ICC)},  May 2017, pp. 1-6.

\bibitem{intermanage'17}
N. Naderializadeh, M. A. Maddah-Ali and A. S. Avestimehr, ``Fundamental limits of cache-aided interference management," \textit{IEEE Trans. Info. Theory}, vol. 63, no. 5, pp. 3092-3107, May 2017.


\bibitem{StoreandLate'17}
F. Xu, M. Tao and K. Liu, ``Fundamental tradeoff between storage and latency in cache-aided wireless interference Networks,'' \textit{IEEE Trans. Info. Theory}, vol. 63, no. 11, pp. 7464-7491, Nov. 2017.


\bibitem{Fog'Overiew}
M. Chiang and T. Zhang, ``Fog and IoT: An overview of research opportunities," \emph{IEEE Internet of Things Journal}, vol. 3, no. 6, pp. 854-864, Dec. 2016.






\bibitem{D2D} M. Ji, G. Caire, and A. F. Molisch, ``Fundamental limits of caching in wireless D2D networks,'' \textit{IEEE Trans. Info. Theory,} vol. 62, no. 2, pp. 849--869, Feb. 2016.

\bibitem{Malak'18}
D. Malak, M. Al-Shalash and J. G. Andrews, ``Spatially correlated content caching for device-to-device communications," \emph{IEEE Transactions on Wireless Communications}, vol. 17, no. 1, pp. 56-70, Jan. 2018.


\bibitem{Karasik'18}
R. Karasik, O. Simeone and S. Shamai, ``Fundamental latency limits for D2D-aided content delivery in fog Wireless networks," \emph{IEEE Int. Symp. Inf. Theory}, 2018, pp. 2461-2465.


\bibitem{Pedersen'19}
J. Pedersen, A. Graelli Amat, I. Andriyanova and F. Br\"annstr\"om, ``Optimizing MDS coded caching in wireless networks with device-to-device communication," \emph{IEEE Transactions on Wireless Communications}, vol. 18, no. 1, pp. 286-295, Jan. 2019.

\bibitem{arXiv'19}
J. Chen, H. Yin, X. You, Y. Geng, and Y. Wu, ``Centralized coded caching with user cooperation," available at arXiv.  2019.

\end{thebibliography}
\end{document}